\documentclass[reqno]{amsart}
\usepackage[all]{xy}
\usepackage[mathscr]{eucal}

\usepackage{amsmath,amsfonts,amssymb,amsthm,epsfig,amscd,comment,latexsym,psfrag}
\usepackage{nicefrac,xspace,tikz}
\usepackage{arydshln}
\usepackage{CJK}
\usepackage{graphicx}
\usepackage{cite}
\usepackage[numbers,sort&compress]{natbib}

\theoremstyle{definition}
\newtheorem{thm}{Theorem}[section]

\newtheorem{lem}[thm]{Lemma}

\newtheorem{eg}[thm]{Example}

\newtheorem{rmk}[thm]{Remark}
\def\footnoterule{\kern 1mm \hrule width 7cm \kern 2.2mm}%

\def\1{{\bf{1}}}
\def\Tr{\mathrm{Tr}}
\date{}

\begin{document}
%\begin{CJK}{GBK}{song}
%\begin{flushleft}
%------------------------Title----------------------------
%\Large\bf{\boldmath{Tighter monogamy and polygamy relations in multi-qubit entanglement systems}}
%\section*{\center\Large\bf Abstract}
\begin{center}
{\LARGE \textbf{Tighter monogamy and polygamy relations of quantum entanglement in multi-qubit systems}}
\end{center}

%\maketitle

%--------------------------Footnote--------------------------
%\footnote{Supported by National Nature Science Foundation of China under
%Grant Nos. 11031005 and 11475116, Beijing Municipal Commission of Education under Grant Nos. KZ201210028032 and KZ201410028033.

%\footnote{$^{\dag}$Corresponding author, E-mail: 2180501016@cnu.edu.cn}

%\vspace{5pt}

%\author{Wenwen Liu}
%\address{$^{1}$School of Mathematical Sciences, Capital Normal University, Beijing 100048, China}
%\email{2180501016@cnu.edu.cn}
%\author{Zifeng Yang}
%\address{School of Mathematical Sciences, Capital Normal University, Beijing 100048, China}
%\email{yangzf@cnu.edu.cn}
%\author{Shaoming Fei}
%\address{School of Mathematical Sciences, Capital Normal University, Beijing 100048, China.
%Max-Planck-Institute for Mathematics in the Sciences, 04103 Leipzig, Germany}
%\address{Max-Planck-Institute for Mathematics in the Sciences, 04103 Leipzig, Germany}

%---------------------------Authors-----------------------------
\begin{center}
\normalsize \rm{} Wen-Wen Liu$^{1^{\dag}}$, Zi-Feng Yang$^{1}$, Shao-Ming Fei$^{1,2}$
\end{center}

\begin{center}
$^{1}$School of Mathematical Sciences, Capital Normal University, Beijing 100048, China\\
$^{2}$Max-Planck-Institute for Mathematics in the Sciences, 04103 Leipzig, Germany\\
$^{\dag}$Corresponding author, E-mail: 2180501016@cnu.edu.cn
\end{center}

\begin{abstract}
 We investigate the monogamy relations related to the concurrence, the entanglement of formation, convex-roof extended negativity, Tsallis-q entanglement and R$\acute{e}$nyi-$\alpha$ entanglement, the polygamy relations related to the entanglement of formation, Tsallis-q entanglement and R$\acute{e}$nyi-$\alpha$ entanglement. Monogamy and polygamy inequalities are obtained for arbitrary multipartite qubit systems, which are proved to be tighter
than the existing ones. Detailed examples are presented.
\end{abstract}


\noindent
{\bf Key words: }{Monogamy relations, Polygamy relations, Concurrence, Entanglement of formation, Negativity, Tsallis-q entanglement, R$\acute{e}$nyi-$\alpha$ entanglement}
%\section {\bf Introduction}
\section{\bf Introduction}
Quantum entanglement is an essential feature of quantum mechanics, which distinguishes the quantum theory from the classical theory \cite{3,4,5,6,7}. The quantification of quantum entanglement is a central issue in quantum information theory \cite{1,2}. As one of the fundamental differences between quantum entanglement and classical correlation, a key property of entanglement is that a quantum system entangled with one of the other systems limits its entanglement with the remaining ones. The monogamy of entanglement (MoE) gives rise to the structures of entanglement in the multipartite setting. Monogamy is also an essential feature allowing for security in quantum key distribution \cite{8}.

For a tripartite quantum state $\rho_{ABC}$, MoE is characterized as $\varepsilon(\rho_{A|BC})\geq\varepsilon(\rho_{AB})+\varepsilon(\rho_{AC})$, where $\rho_{AB}=\Tr_C(\rho_{ABC})$ and $\rho_{AC}=\Tr_B(\rho_{ABC})$ are reduced density matrices, and $\varepsilon$ is an entanglement measure. The well-known concurrence introduced in \cite{9,10} has an explicit expression for arbitrary two-qubit states. Based on this expression, Coffman, Kundu and Wootters \cite{11} derived the famous genuine three-qubit entanglement monotone, three tangle, and conjuctured an inequality for concurrence which describes the monogamy feature of entanglement distribution in a multipartite quantum system. However, such monogamy relations are not always satisfied by any entanglement measures. It has been shown that the squared concurrence $C^2$, and the squared entanglement of formation $E^2$ do satisfies the monogamy relations, while the squared convex-roof extended negativity (CREN) $\widetilde{N}^2$ satisfies the monogamy relations for multiqubit states \cite{12,13,14,15,16}.

Another important concept is the assisted entanglement, which is the amount dual to the bipartite entanglement measure. It has a dually monogamous property in multipartite quantum systems and gives rise to polygamy relations. For a tripartite state $\rho_{ABC}$, the usual polygamy relation is of the form, $\varepsilon_a(\rho_{A|BC})\leq\varepsilon_a(\rho_{AB})+\varepsilon_a(\rho_{AC})$, where $\varepsilon_a$ is the corresponding measure of assisted entanglement associated to $\varepsilon$. Such polygamy inequality has been deeply investigated in recent years, and was generalized to multiqubit systems and classes of higher dimensional quantum systems \cite{17,18,19,20,21,22,23,45}.

Some monogamy and polygamy inequalities related to the $\alpha$th power of entanglement measures have been also proposed. In \cite{24,25,26,27}, it is proved that the $\alpha$th power of concurrence and CREN satisfy the monogamy inequalities in multiqubit systems for $\alpha\geq2$. It has also been shown that the $\alpha$th power of EoF satisfies monogamy relations when $\alpha\geq\sqrt{2}$. Besides, the $\alpha$th power of Tsallis-q entanglement and R$\acute{e}$nyi-$\alpha$ entanglement satisfy monogamy relations when $\alpha\geq1$ for some cases \cite{17,24,25,26,27,28}. The corresponding polygamy relations have also been established \cite{19,20,21,23,29,30}.

In this paper, we present the monogamy inequalities in terms of the concurrence C, entanglement of formation E, convex-roof extended negativity $\widetilde{N}$, Tsallis-q entanglement $T_q$, and R$\acute{e}$nyi-$\alpha$ entanglement $E_{\acute{\alpha}}$, the polygamy inequalities in terms of the entanglement of formation $E_{a}$, Tsallis-q entanglement $T_{aq}$, and R$\acute{e}$nyi-$\alpha$ entanglement $E_{a\acute{\alpha}}$. These inequalities are proved to be tighter than the existing ones.

\section{\bf Tighter monogamy relations for concurrence}
Let $\mathcal{H}_X$ denote a finite-dimensional complex vector space associated to a quantum subsystem $X$. Given a bipartite pure state $|\phi\rangle_{AB}$ in Hilbert space $\mathcal{H}_A\otimes\mathcal{H}_B$, the concurrence is given by
\begin{equation}
C(|\phi\rangle_{AB})=\sqrt{2[1-\Tr(\rho_A^2)]},
\end{equation}
where $\rho_A = \Tr(|\phi\rangle_{AB}\langle\phi|)$ is the reduced density matrix obtained by tracing over the subsystem $B$ \cite{31}-\cite{33}. The concurrence for a bipartite mixed state $\rho_{AB}$ is defined by the convex roof extension,
\begin{equation}
C(\rho_{AB})=\min_{\{p_{i},|\phi_{i}\rangle\}}\sum_{i}p_{i}C(|\phi_{i}\rangle),
\end{equation}
where the minimum is taken over all possible decompositions of $\rho_{AB}=\sum\limits_{i}p_i|\phi_i\rangle\langle\phi_i|$,
with $p_i\geq0$, $\sum\limits_{i}p_i=1$, and $|\phi_i\rangle\in\mathcal{H}_A\otimes\mathcal{H}_B$.
For any $N$-qubit mixed state $\rho_{AB_1 \cdots B_{N-1}}$ in an $N$-qubit system
$\mathcal{H}_A\otimes\mathcal{H}_{B_1}\otimes \cdots \otimes\mathcal{H}_{B_{N-1}}$, the concurrence
$C(\rho_{A|B_1 \cdots B_{N-1}})$ of the state $\rho_{AB_1 \cdots B_{N-1}}$ viewed as a bipartite
state under the partition $A$ and $B_1,B_2,\cdots,B_{N-1}$, satisfies
\begin{equation}\label{b}
C^{\alpha}(\rho_{A|B_1 \cdots B_{N-1}})\geq C^{\alpha}(\rho_{AB_1})+C^{\alpha}(\rho_{AB_2})+\cdots+C^{\alpha}(\rho_{AB_{N-1}}),
\end{equation}
for $\alpha\geq2$, where $\rho_{AB_i}=\Tr_{B_1\cdots B{i-1}B{i+1}\cdots B_{N-1}}(\rho_{AB_1 \cdots B_{N-1}})$ \cite{24}.
The relation (\ref{b}) is improved for $\alpha\geq2$ \cite{25}. If $C(\rho_{AB_i})\geq C(\rho_{A|B_{i+1}\cdots B_{N-1}})$ for $i=1,2,\cdots,m$, and $C(\rho_{AB_j})\leq C(\rho_{A|B_{j+1}\cdots B_{N-1}})$ for $j=m+1,\cdots,N-2$, $1\leq m\leq N-3$, $N\geq4$, then
\begin{eqnarray}\label{c}
C^{\alpha}(\rho_{A|B_1 \cdots B_{N-1}})&\geq& C^{\alpha}(\rho_{AB_1})+(2^{\frac{\alpha}{2}}-1)C^{\alpha}(\rho_{AB_2})+\cdots+(2^{\frac{\alpha}{2}}-1)^{m-1}C^{\alpha}(\rho_{AB_m})\nonumber\\
& &+(2^{\frac{\alpha}{2}}-1)^{m+1}(C^{\alpha}(\rho_{AB_{m+1}})+\cdots+C^{\alpha}(\rho_{AB_{N-2}}))\nonumber\\
& &+(2^{\frac{\alpha}{2}}-1)^mC^{\alpha}(\rho_{AB_{N-1}}).
\end{eqnarray}
The relation (\ref{c}) is further improved for $\alpha\geq2$ as
\begin{eqnarray}\label{d}
C^{\alpha}(\rho_{A|B_1 \cdots B_{N-1}})&\geq&C^{\alpha}(\rho_{AB_1})+(((1+k)^\frac{\alpha}{2}-1)/k^\frac{\alpha}{2})C^{\alpha}(\rho_{AB_2})+\cdots\nonumber\\
& &+(((1+k)^\frac{\alpha}{2}-1)/k^\frac{\alpha}{2})^{m-1}C^{\alpha}(\rho_{AB_m})\nonumber\\
& &+(((1+k)^\frac{\alpha}{2}-1)/k^\frac{\alpha}{2})^{m+1}(C^{\alpha}(\rho_{AB_{m+1}})+\cdots+C^{\alpha}(\rho_{AB_{N-2}}))\nonumber\\
& &+(((1+k)^\frac{\alpha}{2}-1)/k^\frac{\alpha}{2})^mC^{\alpha}(\rho_{AB_{N-1}}),
\end{eqnarray}
with $kC^2(\rho_{AB_i})\geq C^2(\rho_{A|B_{i+1}\cdots B_{N-1}})$ for $i=1,2,\cdots,m$, and $C^2(\rho_{AB_j})\leq kC^2(\rho_{A|B_{j+1}\cdots B_{N-1}})$ for $j=m+1,\cdots,N-2$, $1\leq m\leq N-3$, $N\geq4$, and $0<k\leq1$ \cite{27}.

In the following, we show that these monogamy relations for concurrence can become even tighter under some conditions. For convenience, we denote by $C_{AB_j}=C(\rho_{AB_j})$ for $j=1,2,\cdots,N-1$, and $C_{A|B_1B_2\cdots B_{N-1}}=C(\rho_{A|B_1 B_2\cdots B_{N-1}})$. We first introduce the following lemma.

\begin{lem}\label{dd}
For any non-negative real number $x$ and $y$ satisfying $0\leq y\leq x$, and real numbers $t$ and $s$ satisfying $t\geq1$, $0\leq s\leq1$, we have
\begin{equation}\label{e}
(1+x)^t-x^t\geq(1+y)^t-y^t,
\end{equation}
\begin{equation}\label{f}
(1+x)^s-x^s\leq(1+y)^s-y^s.
\end{equation}
\end{lem}

\begin{proof}
Let $g(x,t)=(1+x)^t-x^t$. Since $\frac{\partial g(x,t)}{\partial x}=t[(1+x)^{t-1}-x^{t-1}]\geq0$, the function $g(x,t)$ is increasing with respect to $x$. As $y\leq x$, $g(y,t)\leq g(x,t)$, we get the inequality (\ref{e}). Similar to the proof of inequality (\ref{e}), we can obtain the inequality (\ref{f}).
\end{proof}

For any $2\otimes2\otimes2^{n-2}$ mixed state $\rho_{ABC}\in\mathcal{H}_A\otimes\mathcal{H}_B\otimes\mathcal{H}_C$, we have from
relation (\ref{b})
$$
C_{A|BC}^2\geq C_{AB}^2+C_{AC}^2.
$$
Therefore, there exists $\mu\ge 1$ such that
\begin{equation}\label{mu}
C_{A|BC}^2\geq C_{AB}^2+\mu C_{AC}^2.
\end{equation}

\begin{lem}\label{ee}
 Let $l\geq1$ be a real number. For any $2\otimes2\otimes2^{n-2}$ mixed state $\rho_{ABC}\in\mathcal{H}_A\otimes\mathcal{H}_B\otimes\mathcal{H}_C$, if $C_{AB}^2\geq lC_{AC}^2$, we have
\begin{eqnarray}
C_{A|BC}^\alpha\geq C_{AB}^\alpha+((\mu+l)^\frac{\alpha}{2}-l^\frac{\alpha}{2})C_{AC}^\alpha,
\end{eqnarray}
for all $\alpha\geq2$.
\end{lem}

\begin{proof}
By straightforward calculation, we have
\begin{eqnarray}
C_{A|BC}^\alpha&=&(C_{A|BC}^2)^{\frac{\alpha}{2}}\geq(C_{AB}^2+\mu C_{AC}^2)^{\frac{\alpha}{2}}\nonumber\\
&=&\mu^{\frac{\alpha}{2}}C_{AC}^{\alpha}[(\mu^{-1}({C_{AB}^2}/{C_{AC}^2})+1)^{\frac{\alpha}{2}}-(\mu^{-1}({C_{AB}^2}/{C_{AC}^2}))^{\frac{\alpha}{2}}]+C_{AB}^{\alpha}\nonumber\\
&\geq&\mu^\frac{\alpha}{2}C_{AC}^{\alpha}[(\mu^{-1}l+1)^{\frac{\alpha}{2}}-(\mu^{-1}l)^{\frac{\alpha}{2}}]+C_{AB}^{\alpha}\nonumber\\
&=&[(l+\mu)^{\frac{\alpha}{2}}-l^{\frac{\alpha}{2}}]C_{AC}^{\alpha}+C_{AB}^{\alpha},
\end{eqnarray}
where the second inequality is due to Lemma \ref{dd}. We can also see that if $C_{AB}=0$, then $C_{AC}=0$, and the lower bound becomes trivially zero.
\end{proof}

For multiqubit systems, we have the following theorems.
\begin{thm}\label{ff}
Let $\mu_r\geq1$ and $l_r\geq1$ be real numbers, $1\le r\le N-2$. For any $N$-qubit mixed state $\rho_{AB_1\cdots B_{N-1}}\in\mathcal{H}_A\otimes\mathcal{H}_{B_1}\otimes\cdots\otimes\mathcal{H}_{B_{N-1}}$, if $C_{AB_i}^2\geq l_iC_{A|B_{i+1}\cdots B_{N-1}}^2$, $C_{A|B_i\cdots B_{N-1}}^2\geq C_{AB_i}^2+\mu_iC_{A|B_{i+1}\cdots B_{N-1}}^2$ for $i=1,2,\cdots,m$, and $C_{A|B_{j+1}\cdots B_{N-1}}^2\geq l_jC_{AB_j}^2$, $C_{A|B_j\cdots B_{N-1}}^2\geq\mu_jC_{AB_j}^2+C_{A|B_{j+1}\cdots B_{N-1}}^2$ for $j=m+1,
\cdots,N-2$, $1\leq m\leq N-3,N\geq4$, we have
\begin{eqnarray}
C_{A|B_1\cdots B_{N-1}}^\alpha&\geq&C_{AB_1}^\alpha+\mathcal{K}_1C_{AB_2}^\alpha+\cdots+\mathcal{K}_1\cdots\mathcal{K}_{m-1}C_{AB_m}^\alpha\nonumber\\
& &+\mathcal{K}_1\cdots\mathcal{K}_{m}(\mathcal{K}_{m+1}C_{AB_{m+1}}^\alpha+\cdots+\mathcal{K}_{N-2}C_{AB_{N-2}}^\alpha)\nonumber\\
& &+\mathcal{K}_1\cdots\mathcal{K}_{m}C_{AB_{N-1}}^\alpha
\end{eqnarray}
for all $\alpha\geq2$, where $\mathcal{K}_r=(\mu_r+l_r)^{\frac{\alpha}{2}}-l_r^{\frac{\alpha}{2}}$ with $1\le r\le N-2$.
\end{thm}

\begin{proof}
From Lemma \ref{ee}, we have
\begin{eqnarray}\label{x}
C_{A|B_1\cdots B_{N-1}}^\alpha&\geq&C_{AB_1}^\alpha+\mathcal{K}_1C_{A|B_2\cdots B_{N-1}}^\alpha\nonumber\\
&\geq&C_{AB_1}^\alpha+\mathcal{K}_1C_{AB_2}^\alpha+\mathcal{K}_1\mathcal{K}_2C_{A|B_3\cdots B_{N-1}}^\alpha\geq\cdots\nonumber\\
&\geq&C_{AB_1}^\alpha+\mathcal{K}_1C_{AB_2}^\alpha+\cdots+\mathcal{K}_1\cdots\mathcal{K}_{m-1}C_{AB_m}^\alpha\nonumber\\
& &+\mathcal{K}_1\cdots\mathcal{K}_{m}C_{A|B_{m+1}\cdots B_{N-1}}^\alpha.
\end{eqnarray}
Since $C_{A|B_{j+1}\cdots B_{N-1}}^2\geq l_jC_{AB_j}^2$, $C_{A|B_j\cdots B_{N-1}}^2\geq\mu_jC_{AB_j}^2+C_{A|B_{j+1}\cdots B_{N-1}}^2$ for $j=m+1,
\cdots,N-2$, we get
\begin{eqnarray}\label{y}
C_{A|B_{m+1}\cdots B_{N-1}}^\alpha&\geq&\mathcal{K}_{m+1}C_{AB_{m+1}}^\alpha+C_{A|B_{m+2}\cdots B_{N-1}}^\alpha\nonumber\\
&\geq&\mathcal{K}_{m+1}C_{AB_{m+1}}^\alpha+\mathcal{K}_{m+2}C_{AB_{m+2}}^\alpha+C_{A|B_{m+3}\cdots B_{N-1}}^\alpha\geq\cdots\nonumber\\
&\geq&\mathcal{K}_{m+1}C_{AB_{m+1}}^\alpha+\mathcal{K}_{m+2}C_{AB_{m+2}}^\alpha+\cdots+\nonumber\\
& &+\mathcal{K}_{N-2}C_{AB_{N-2}}^\alpha+C_{AB_{N-1}}^\alpha.
\end{eqnarray}
Combining (\ref{x}) and (\ref{y}), we complete the proof.
\end{proof}

An immediate corollary of Theorem \ref{ff}, we have in particular,

\begin{thm}
Let $\mu_r\geq1$ and $l_r\geq1$ be real numbers, $1\le r\le N-2$.
For any $N$-qubit mixed state $\rho_{AB_1\cdots B_{N-1}}\in\mathcal{H}_A\otimes\mathcal{H}_{B_1}\otimes\cdots\otimes\mathcal{H}_{B_{N-1}}$, if $C_{AB_i}^2\geq l_iC_{A|B_{i+1}\cdots B_{N-1}}^2$, $C_{A|B_i\cdots B_{N-1}}^2\geq C_{AB_i}^2+\mu_iC_{A|B_{i+1}\cdots B_{N-1}}^2$ for all $i=1,2,\cdots,N-2$, then we have
\begin{eqnarray}
C_{A|B_1\cdots B_{N-1}}^\alpha\geq C_{AB_1}^\alpha+\mathcal{K}_1C_{AB_2}^\alpha
+\cdots+\mathcal{K}_1\cdots\mathcal{K}_{N-2}C_{AB_{N-1}}^\alpha,
\end{eqnarray}
for all $\alpha\geq2$, where $\mathcal{K}_r=(\mu_r+l_r)^{\frac{\alpha}{2}}-l_r^{\frac{\alpha}{2}}$ with $1\le r\le N-2$.
\end{thm}

\begin{rmk}
Since
\begin{equation}\label{z}
(\mu+l)^{\frac{\alpha}{2}}-l^{\frac{\alpha}{2}}\geq(1+l)^{\frac{\alpha}{2}}-l^{\frac{\alpha}{2}}\geq(2)^{\frac{\alpha}{2}}-l
\end{equation}
for $\alpha\geq2$, $\mu\geq1$ and $l\geq1$, we have $(1+l)^{\frac{\alpha}{2}}-l^{\frac{\alpha}{2}}=\frac{(1+k)^\frac{\alpha}{2}-1}{k^\frac{\alpha}{2}}$ if $l=\frac{1}{k}$ with $0 \textless k\leq1$. In (\ref{z}) the first equality holds when $\mu=1$ and the second equality holds when $l=1$. For given $l$, the bigger the $\mu$ is, the tighter the inequality in Theorem \ref{ff} is. Therefore, our new monogamy relation for concurrence is better than the ones in \cite{25,27}.
\end{rmk}

\begin{eg}\label{a}
Let us consider the three-qubit state $|\phi\rangle_{ABC}$ in the generalized Schmidt decomposition from \cite{36,37},
\begin{equation}
|\phi\rangle_{ABC}=\lambda_0|000\rangle+\lambda_1e^{i\varphi}|100\rangle+\lambda_2|101\rangle+\lambda_3|110\rangle+\lambda_4|111\rangle,
\end{equation}
where $\lambda_i\geq0$, $i=0,1,\cdots,4$, and $\sum\limits_{i=0}^{4}\lambda_i^2=1$.
One gets $C_{A|BC}=2\lambda_0\sqrt{\lambda_2^2+\lambda_3^2+\lambda_4^2}$, $C_{AB}=2\lambda_0\lambda_2$ and $C_{AC}=2\lambda_0\lambda_3$. Setting $\lambda_0=\lambda_3=\lambda_4={1}/{\sqrt{5}}$, $\lambda_2=\sqrt{{2}/{5}}$ and $\lambda_1=0$, we have $C_{A|BC}={4}/{5}$, $C_{AB}={2\sqrt{2}}/{5}$ and $C_{AC}={2}/{5}$. Therefore,
\begin{equation}\label{g}
C_{AB}^\alpha+(2^\frac{\alpha}{2}-1)C_{AC}^\alpha=(2\sqrt{2}/{5})^\alpha+(2^\frac{\alpha}{2}-1)(2/{5})^\alpha,
\end{equation}
\begin{equation}\label{h}
C_{AB}^\alpha+(((1+k)^\frac{\alpha}{2}-1)/{k^\frac{\alpha}{2}})C_{AC}^\alpha=({2\sqrt{2}}/{5})^\alpha+({(1+k)^\frac{\alpha}{2}-1})/{k^\frac{\alpha}{2}})({2}/{5})^\alpha,
\end{equation}
\begin{equation}\label{i}
C_{AB}^\alpha+((\mu+l)^\frac{\alpha}{2}-l^\frac{\alpha}{2})C_{AC}^\alpha=({2\sqrt{2}}/{5})^\alpha+((\mu+l)^\frac{\alpha}{2}-l^\frac{\alpha}{2})({2}/{5})^\alpha.
\end{equation}
When $k=0.5$ the lower bound (\ref{h}) gives the best result. When $l=\frac{1}{k}=2,\mu=1$ the lower bound (\ref{i}) gives the same result as (\ref{h}). But when $l=\frac{1}{k}=2$ and $1\textless\mu\leq2$, the lower bound (\ref{i}) is better than (\ref{h}). It can be seen that our result is better than the result (\ref{h}) in \cite{27} for $\alpha\geq2$, hence better than (\ref{g}) given in \cite{25}, see Figure 1.
\end{eg}
\begin{figure}[h]
\begin{center}
\includegraphics[width=8cm,clip]{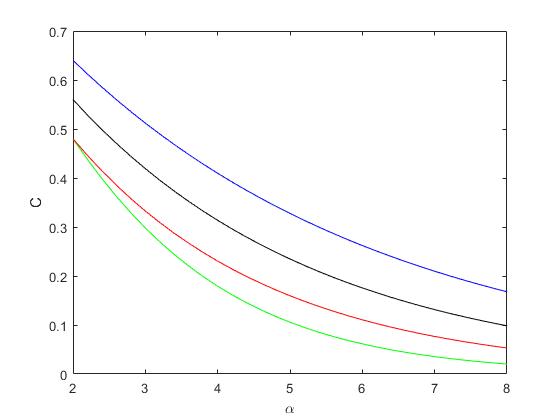}
\caption{From top to bottom, the first curve represents the concurrence of $|\phi\rangle_{A|BC}$ in Example \ref{a}, the third and fourth curves represent the lower bounds from \cite{27} and \cite{25}, respectively. The second curve represents the lower bound from our result.}	
\end{center}
\end{figure}

\section{\bf Tighter monogamy and polygamy relations for EoF}

Let $\mathcal{H}_A$ and $\mathcal{H}_B$ be $m$ and $n$ dimensional $(m\leq n)$ vector space, respectively. The entanglement of formation (EoF) of a pure state $|\phi\rangle\in\mathcal{H}_A\otimes\mathcal{H}_B$ is defined by $E(|\phi\rangle)=S(\rho_A)$,
where $\rho_A=\Tr(|\phi\rangle\langle\phi|)$ and $S(\rho)=-\Tr(\rho\log_2\rho)$ \cite{38,39}. For a bipartite mixed state $\rho_{AB}\in\mathcal{H}_A\otimes\mathcal{H}_B$, the EoF is given by
\begin{equation}
E(\rho_{AB})=\min\limits_{\{p_i,|\phi_i\rangle\}}\sum\limits_{i}p_iE(|\phi_i\rangle),
\end{equation}
with the minimum taking over all possible pure state decomposition of $\rho_{AB}$.

Denote by $f(x)=H(\frac{1+\sqrt{1-x}}{2})$, where $H(x)=-x\log_2(x)-(1-x)\log_2(1-x)$. It is obvious that $f(x)$ is a monotonically increasing function for $0\leq x\leq1$ which satisfies
\begin{equation}\label{m}
f^{\sqrt{2}}(x^2+y^2)\geq f^{\sqrt{2}}(x^2)+f^{\sqrt{2}}(y^2),
\end{equation}
\begin{equation}
f(x^2+y^2)\leq f(x^2)+f(y^2),
\end{equation}
where $f^{\sqrt{2}}(x^2+y^2)=[f(x^2+y^2)]^{\sqrt{2}}$. It is showed in \cite{10} that $E(|\varphi\rangle)=f(C^2(|\varphi\rangle))$ for $2\otimes m $ ($m\geq2$) pure state $|\varphi\rangle$, and $E(\rho)=f(C^2(\rho))$ for two-qubit mixed state $\rho$.

EoF does not satisfy the inequality $E_{A|BC}\geq E_{AB}+E_{AC}$ \cite{11}. In \cite{40} it is shown that EoF is a monotonic function: $E^2(C_{A|B_1B_2\cdots B_{N-1}}^2)\geq E^2(\sum\limits_{i=1}^{N-1}C_{AB_i}^2)$. It is further proved that for $N$-qubit systems, $E_{A|B_1B_2\cdots B_{N-1}}^\alpha\geq E_{AB_1}^\alpha+E_{AB_2}^\alpha+\cdots+E_{AB_{N-1}}^\alpha$ for $\alpha\geq\sqrt{2}$, where $E_{A|B_1B_2\cdots B_{N-1}}$ is the EoF of $\rho$ in bipartite partition $A|B_1B_2\cdots B_{N-1}$, and $E_{AB_i}$, $i=1,2,\cdots,N-1$, is the EoF of the bipartite states $\rho_{AB_i}=\Tr_{B_1B_2\cdots B_{i-1}B_{i+1}\cdots B_{N-1}}(\rho)$ \cite{24}.

For $N$-qubit systems, the following monogamy relation has been obtained,
\begin{eqnarray}
E^{\alpha}(\rho_{A|B_1\cdots B_{N-1}})&\geq&E^{\alpha}(\rho_{AB_1})+(2^{\frac{\alpha}{\sqrt{2}}}-1)E^{\alpha}(\rho_{AB_2})+\cdots+(2^{\frac{\alpha}{\sqrt{2}}}-1)^{m-1}E^{\alpha}(\rho_{AB_m})\nonumber\\
& &+(2^{\frac{\alpha}{\sqrt{2}}}-1)^{m+1}(E^{\alpha}(\rho_{AB_{m+1}})+\cdots+E^{\alpha}(\rho_{AB_{N-2}}))\nonumber\\
& &+(2^{\frac{\alpha}{\sqrt{2}}}-1)^mE^{\alpha}(\rho_{AB_{N-1}})
\end{eqnarray}
for $\alpha\geq\sqrt{2}$, with the conditions $C(\rho_{AB_i})\geq C(\rho_{A|B_{i+1}\cdots B_{N-1}})$ for $i=1,2,\cdots,m$, and $C(\rho_{AB_j})\leq C(\rho_{A|B_{j+1}\cdots B_{N-1}})$ for $j=m+1,\cdots,N-2$, $1\leq m\leq N-3$, $N\geq4$ \cite{25}.
The inequality (\ref{l}) is a further improvement \cite{27} as
\begin{eqnarray}\label{l}
E^{\alpha}(\rho_{A|B_1 \cdots B_{N-1}})&\geq&E^{\alpha}(\rho_{AB_1})+(((1+k)^{\frac{\alpha}{\sqrt{2}}}-1)/{k^{\frac{\alpha}{\sqrt{2}}}})E^{\alpha}(\rho_{AB_2})+\cdots\nonumber\\
& &+(((1+k)^{\frac{\alpha}{\sqrt{2}}}-1)/{k^{\frac{\alpha}{\sqrt{2}}}})^{m-1}E^{\alpha}(\rho_{AB_m})\nonumber\\
& &+(((1+k)^{\frac{\alpha}{\sqrt{2}}}-1)/{k^{\frac{\alpha}{\sqrt{2}}}})^{m+1}(E^{\alpha}(\rho_{AB_{m+1}})+\cdots+E^{\alpha}(\rho_{AB_{N-2}}))\nonumber\\
& &+(((1+k)^{\frac{\alpha}{\sqrt{2}}}-1)/{k^{\frac{\alpha}{\sqrt{2}}}})^mE^{\alpha}(\rho_{AB_{N-1}}),
\end{eqnarray}
for $\alpha\geq\sqrt{2}$, with $kE^{\sqrt{2}}(\rho_{AB_i})\geq E^{\sqrt{2}}(\rho_{A|B_{i+1}\cdots B_{N-1}})$ for $i=1,2,\cdots,m$, and $E^{\sqrt{2}}(\rho_{AB_j})\leq kE^{\sqrt{2}}(\rho_{A|B_{j+1}\cdots B_{N-1}})$ for $j=m+1,\cdots,N-2$, $1\leq m\leq N-3$, $N\geq4$ and $0<k\leq1$.

The corresponding entanglement of assistance (EoA) is defined in terms of the entropy of entanglement for a tripartite pure state $|\phi\rangle_{ABC}$,
\begin{equation}
E_a(|\phi\rangle_{ABC})\equiv E_a(\rho_{AB})=\max\limits_{\{p_i,|\phi_i\rangle\}}\sum\limits_{i}p_iE(|\phi_i\rangle),
\end{equation}
where the maximum is taken over all possible pure state decompositions of $\rho_{AB}=\Tr_C(|\phi\rangle_{ABC}\langle\phi|)=\sum\limits_ip_i|\phi_i\rangle_{AB}\langle\phi_i|$ with $p_i\geq0$ and $\sum\limits_{i}p_i=1$ \cite{41}. For an arbitrary dimensional multipartite quantum state $\rho_{AB_1B_2\cdots B{N-1}}$, a general polygamy inequality of multipartite quantum entanglement was established in \cite{21},
\begin{equation}
E_a(\rho_{A|B_1B_2\cdots B{N-1}})\leq\sum\limits_{i}^{N-1}E_a(\rho_{AB_i}).
\end{equation}

In the following, we show that these monogamy and polygamy relations for EoF can become even tighter under some conditions. For convenience, we denote by $E_{AB_j}=E(\rho_{AB_j})$ for $j=1,2,\cdots,N-1$, and $E_{A|B_1B_2\cdots B_{N-1}}=E(\rho_{A|B_1 B_2\cdots B_{N-1}})$.

\begin{thm}\label{gg}
Let $\mu_r\geq1$ and $l_r\geq1$ be real numbers, $1\le r\le N-2$. For any $N$-qubit mixed state $\rho_{AB_1\cdots B_{N-1}}\in\mathcal{H}_A\otimes\mathcal{H}_{B_1}\otimes\cdots\otimes\mathcal{H}_{B_{N-1}}$, if $E^{\sqrt{2}}_{AB_i}\geq l_iE^{\sqrt{2}}_{A|B_{i+1}\cdots B_{N-1}}$, $E^{\sqrt{2}}_{A|B_i\cdots B_{N-1}}\geq E^{\sqrt{2}}_{AB_i}+\mu_iE^{\sqrt{2}}_{A|B_{i+1}\cdots B_{N-1}}$ for $i=1,2,\cdots,m$, and $E^{\sqrt{2}}_{A|B_{j+1}\cdots B_{N-1}}\geq l_jE^{\sqrt{2}}_{AB_j}$, $E^{\sqrt{2}}_{A|B_j\cdots B_{N-1}}\geq\mu_jE^{\sqrt{2}}_{AB_j}+E^{\sqrt{2}}_{A|B_{j+1}\cdots B_{N-1}}$ for $j=m+1,
\cdots,N-2$, $1\leq m\leq N-3,N\geq4$, then
\begin{eqnarray}
E_{A|B_1\cdots B_{N-1}}^\alpha&\geq& E_{AB_1}^\alpha+\mathcal{K}_1E_{AB_2}^\alpha+\cdots+\mathcal{K}_1\cdots\mathcal{K}_{m-1}E_{AB_m}^\alpha\nonumber\\
& &+\mathcal{K}_1\cdots\mathcal{K}_{m}(\mathcal{K}_{m+1}E_{AB_{m+1}}^\alpha+\cdots+\mathcal{K}_{N-2}E_{AB_{N-2}}^\alpha)\nonumber\\
& &+\mathcal{K}_1\cdots\mathcal{K}_{m}E_{AB_{N-1}}^\alpha
\end{eqnarray}
for all $\alpha\geq\sqrt{2}$, where $\mathcal{K}_r=(\mu_r+l_r)^{\frac{\alpha}{\sqrt{2}}}-l_r^{\frac{\alpha}{\sqrt{2}}}$
with $1\le r\le N-2$.
\end{thm}

\begin{proof}
Consider $\alpha\geq\sqrt{2}$ and $f^{\sqrt{2}}(x^2)\geq lf^{\sqrt{2}}(y^2)$. Due to inequality (\ref{m}), there exists $\mu\geq1$ such that $f^{\sqrt{2}}(x^2+y^2)\geq f^{\sqrt{2}}(x^2)+\mu f^{\sqrt{2}}(y^2)$. Hence we have
\begin{eqnarray}\label{o}
f^{\alpha}(x^2+y^2)&=&[f^{\sqrt{2}}(x^2+y^2)]^{\frac{\alpha}{\sqrt{2}}}\geq[f^{\sqrt{2}}(x^2)+\mu f^{\sqrt{2}}(y^2)]^{\frac{\alpha}{\sqrt{\sqrt{2}}}}\nonumber\\
&=&\mu^{\frac{\alpha}{\sqrt{2}}}f^{\alpha}(y^2)[(\mu^{-1}({f^{\sqrt{2}}(x^2)}/{f^{\sqrt{2}}(y^2)})+1)^{\frac{\alpha}{\sqrt{2}}}-(\mu^{-1}({f^{\sqrt{2}}(x^2)}/{f^{\sqrt{2}}(y^2)}))^{\frac{\alpha}{\sqrt{2}}}]+f^{\alpha}(x^2)\nonumber\\
&\geq&\mu^{\frac{\alpha}{\sqrt{2}}}f^{\alpha}(y^2)[(\mu^{-1}l+1)^{\frac{\alpha}{\sqrt{2}}}-(\mu^{-1}l)^{\frac{\alpha}{\sqrt{2}}}]+f^{\alpha}(x^2)\nonumber\\
&=&[(\mu+l)^{\frac{\alpha}{\sqrt{2}}}-l^{\frac{\alpha}{\sqrt{2}}}]f^{\alpha}(y^2)+f^{\alpha}(x^2),
\end{eqnarray}
where the second inequality is obtained from inequality (\ref{e}).
Let $\rho=\sum_ip_i|\phi_i\rangle\langle\phi|\in\mathcal{H}_A\otimes\mathcal{H}_{B_1}\otimes\cdots\otimes\mathcal{H}_{B_{N-1}}$ be the optimal decomposition of $E_{A|B_1\cdots B_{N-1}}(\rho)$ for the $N$-qubit mixed state $\rho$. Then from \cite{25}
\begin{equation}\label{n}
E_{A|B_1\cdots B_{N-1}}\geq f(C_{A|B_1\cdots B_{N-1}}^2).
\end{equation}
Therefore,
\begin{eqnarray}
E_{A|B_1\cdots B_{N-1}}^\alpha&\geq&f^\alpha(C_{A|B_1\cdots B_{N-1}}^2)\nonumber\\
&\geq&f^\alpha(C_{AB_1}^2)+\mathcal{K}_1f^\alpha(C_{AB_2}^2)+\cdots+\mathcal{K}_1\cdots\mathcal{K}_{m-1}f^\alpha(C_{AB_m}^2)\nonumber\\
& &+\mathcal{K}_1\cdots\mathcal{K}_{m}(\mathcal{K}_{m+1}f^\alpha(C_{AB_{m+1}}^2)+\cdots+\mathcal{K}_{N-2}f^\alpha(C_{AB_{N-2}}^2))\nonumber\\
& &+\mathcal{K}_1\cdots\mathcal{K}_{m}f^\alpha(C_{AB_{N-1}}^2)\nonumber\\
&=&E^\alpha_{AB_1}+\mathcal{K}_1E^\alpha_{AB_2}+\cdots+\mathcal{K}_1\cdots\mathcal{K}_{m-1}E^\alpha_{AB_m}\nonumber\\
& &+\mathcal{K}_1\cdots\mathcal{K}_{m}(\mathcal{K}_{m+1}E^\alpha_{AB_{m+1}}+\cdots+\mathcal{K}_{N-2}E^\alpha_{AB_{N-2}})\nonumber\\
& &+\mathcal{K}_1\cdots\mathcal{K}_{m}E^\alpha_{AB_{N-1}},
\end{eqnarray}
where the first inequality is due to (\ref{n}), the second inequality is obtained, similar to the proof of Theorem \ref{ff}, by using inequality (\ref{o}). The last equality holds since for any $2\otimes2$ quantum state $\rho_{AB_i}$, $E(\rho_{AB_i})=f[C^2(\rho_{AB_i})]$.
\end{proof}

In particular, we have

\begin{thm}
Let $\mu_r\geq1$ and $l_r\geq1$ be real numbers, $1\le r\le N-2$. For any $N$-qubit mixed state $\rho_{AB_1\cdots B_{N-1}}\in\mathcal{H}_A\otimes\mathcal{H}_{B_1}\otimes\cdots\otimes\mathcal{H}_{B_{N-1}}$, if $E^{\sqrt{2}}_{AB_i}\geq l_iE^{\sqrt{2}}_{A|B_{i+1}\cdots B_{N-1}}$, $E^{\sqrt{2}}_{A|B_i\cdots B_{N-1}}\geq E^{\sqrt{2}}_{AB_i}+\mu_iE^{\sqrt{2}}_{A|B_{i+1}\cdots B_{N-1}}$ for $i=1,2,\cdots,N-2$, then
\begin{eqnarray}
E_{A|B_1\cdots B_{N-1}}^\alpha\geq E_{AB_1}^\alpha+\mathcal{K}_1E_{AB_2}^\alpha+\cdots+\mathcal{K}_1\cdots\mathcal{K}_{N-2}E_{AB_{N-1}}^\alpha
\end{eqnarray}
for all $\alpha\geq\sqrt{2}$, where $\mathcal{K}_r=(\mu_r+l_r)^{\frac{\alpha}{\sqrt{2}}}-l_r^{\frac{\alpha}{\sqrt{2}}}$
with $1\le r\le N-2$.
\end{thm}

\begin{rmk}
Since
$(\mu+l)^{\frac{\alpha}{\sqrt{2}}}-l^{\frac{\alpha}{\sqrt{2}}}\geq(1+l)^{\frac{\alpha}{\sqrt{2}}}-l^{\frac{\alpha}{\sqrt{2}}}\geq(2)^{\frac{\alpha}{\sqrt{2}}}-l$,
where $\alpha\geq\sqrt{2}$, $\mu\geq1$, $l\geq1$, we have $(1+l)^{\frac{\alpha}{\sqrt{2}}}-l^{\frac{\alpha}{\sqrt{2}}}=((1+k)^\frac{\alpha}{\sqrt{2}}-1)/{k^\frac{\alpha}{\sqrt{2}}}$
when $l=\frac{1}{k}$ with $0\textless k\leq1$. The first equality holds when $\mu=1$ and the second equality holds when $l=1$. For given $l$, the bigger the $\mu$ is, the tighter the inequality in Theorem \ref{gg} is. Hence, our new monogamy relation for EoF is better than the ones in \cite{25,27}.
\end{rmk}

\begin{eg}
Let us again consider the three-qubit state $|\phi\rangle_{ABC}$ defined in Example \ref{a} with $\lambda_0=\lambda_3=\lambda_4={1}/{\sqrt{5}}$, $\lambda_2=\sqrt{{2}/{5}}$ and $\lambda_1=0$. We have
\begin{eqnarray*}
E_{A|BC}&=&-(4/5)\log_2(4/5)-(1/5)\log_2(1/5)\approx0.721928,\\
E_{AB}&=&-((5+\sqrt{17})/10)\log_2((5+\sqrt{17})/10)-((5-\sqrt{17})/10)\log_2((5-\sqrt{17})/10)\approx0.428710,\\
E_{AC}&=&-((5+\sqrt{21})/10)\log_2((5+\sqrt{21})/10)-((5-\sqrt{21})/10)\log_2((5-\sqrt{21})/10)\approx0.250225.
\end{eqnarray*}
Thus,
\begin{eqnarray}
E_{AB}^\alpha+(2^\frac{\alpha}{\sqrt{2}}-1)E_{AC}^\alpha&=&(0.428710)^\alpha+(2^\frac{\alpha}{\sqrt{2}}-1)(0.250225)^\alpha,\\
E_{AB}^\alpha+(\frac{(1+k)^\frac{\alpha}{\sqrt{2}}-1}{k^\frac{\alpha}{\sqrt{2}}})E_{AC}^\alpha&=&(0.428710)^\alpha+(\frac{(1+k)^\frac{\alpha}{\sqrt{2}}-1}{k^\frac{\alpha}{\sqrt{2}}})(0.250225)^\alpha,\\
E_{AB}^\alpha+((\mu+l)^\frac{\alpha}{\sqrt{2}}-l^\frac{\alpha}{\sqrt{2}})E_{AC}^\alpha&=&(0.428710)^\alpha+((\mu+l)^\frac{\alpha}{\sqrt{2}}-l^\frac{\alpha}{\sqrt{2}})(0.250225)^\alpha.
\end{eqnarray}
We see that our result is better than the one in \cite{25,27}, see Figure 2.
\end{eg}
\begin{figure}[h]
\begin{center}
\includegraphics[width=8cm,clip]{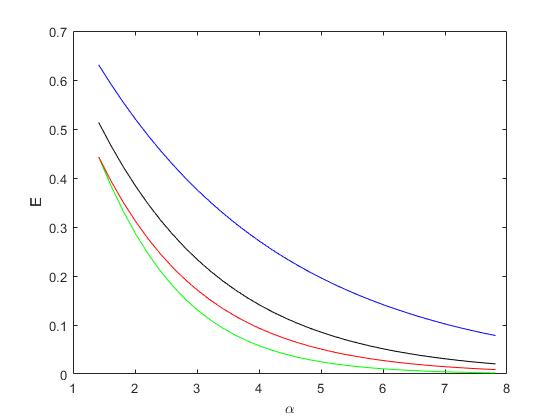}
\caption{From top to bottom, the first curve represents the EoF $E(|\phi\rangle_{A|BC})$, the third curve and the fourth curves represent the lower bounds
from \cite{27} and \cite{25}, respectively, the second curve represents the lower bound from our result.}	
\end{center}
\end{figure}

We can also provide tighter polygamy relations for the entanglement of assistance.

\begin{thm}
Let $0\textless\mu_r\leq1$ and $l_r\geq1$ be real numbers, $1\le r\le N-2$. For any $N$-qubit mixed state $\rho_{AB_1\cdots B_{N-1}}\in\mathcal{H}_A\otimes\mathcal{H}_{B_1}\otimes\cdots\otimes\mathcal{H}_{B_{N-1}}$, if $E_{aAB_i}\geq l_iE_{aA|B_{i+1}\cdots B_{N-1}}$, $E_{aA|B_i\cdots B_{N-1}}\leq E_{aAB_i}+\mu_iE_{aA|B_{i+1}\cdots B_{N-1}}$ for $i=1,2,\cdots,m$, and $E_{aA|B_{j+1}\cdots B_{N-1}}\geq l_jE_{aAB_j}$, $E_{aA|B_j\cdots B_{N-1}}\leq\mu_jE_{aAB_j}+E_{aA|B_{j+1}\cdots B_{N-1}}$ for $j=m+1,
\cdots,N-2$, $1\leq m\leq N-3,N\geq4$, we have
\begin{eqnarray}
E_{aA|B_1\cdots B_{N-1}}^\alpha&\leq&E_{aAB_1}^\alpha+\mathcal{K}_1E_{aAB_2}^\alpha+\cdots+\mathcal{K}_1\cdots\mathcal{K}_{m-1}E_{aAB_m}^\alpha\nonumber\\
& &+\mathcal{K}_1\cdots\mathcal{K}_{m}(\mathcal{K}_{m+1}E_{aAB_{m+1}}^\alpha+\cdots+\mathcal{K}_{N-2}E_{aAB_{N-2}}^\alpha)\nonumber\\
& &+\mathcal{K}_1\cdots\mathcal{K}_{m}E_{aAB_{N-1}}^\alpha
\end{eqnarray}
for all $0\leq\alpha\leq1$, where $\mathcal{K}_r=(\mu_r+l_r)^{\alpha}-l_r^{\alpha}$ with $1\le r\le N-2$.
\end{thm}

Particularly, we have

\begin{thm}
Let $0\textless\mu_r\leq1$ and $l_r\geq1$ be real numbers, $1\le r\le N-2$. For any $N$-qubit mixed state $\rho_{AB_1\cdots B_{N-1}}\in\mathcal{H}_A\otimes\mathcal{H}_{B_1}\otimes\cdots\otimes\mathcal{H}_{B_{N-1}}$, if $E_{aAB_i}\geq l_iE_{aA|B_{i+1}\cdots B_{N-1}}$, $E_{aA|B_i\cdots B_{N-1}}\leq E_{aAB_i}+\mu_iE_{aA|B_{i+1}\cdots B_{N-1}}$ for $i=1,2,\cdots,N-2$, then
\begin{eqnarray}
E_{aA|B_1\cdots B_{N-1}}^\alpha\leq E_{aAB_1}^\alpha+\mathcal{K}_1E_{aAB_2}^\alpha+\cdots+\mathcal{K}_1\cdots\mathcal{K}_{N-2}E_{aAB_{N-1}}^\alpha
\end{eqnarray}
for all $0\leq\alpha\leq1$, where $\mathcal{K}_r=(\mu_r+l_r)^{\alpha}-l_r^{\alpha}$ with $1\le r\le N-2$.
\end{thm}

\section{\bf Tighter monogamy relations for negativity}
Another well-known quantifier of bipartite entanglement is the negativity, which is based on the positive partial transposition (PPT) criterion.
For a bipartite state $\rho_{AB}$ in $\mathcal{H}_A\otimes\mathcal{H}_B$ the negativity is given by
$N(\rho_{AB})=(\left\|\rho_{AB}^{T_A}\right\|-1)/2$, where $\rho_{AB}^{T_A}$ is the partial transpose with respect to the subsystem $A$, and $\left\|X\right\|$ denotes the trace norm of X, i.e., $\left\|X\right\|=\sqrt{XX^{\dag}}$ \cite{42}. For the purposes of discussion, we use the definition of negativity as $\left\|\rho_{AB}^{T_A}\right\|-1$. For a bipartite mixed state $\rho_{AB}$, the convex-roof extended negativity (CREN) is defined by
\begin{equation}
\widetilde{N}(\rho_{AB})=\min\limits_{\{p_i,|\phi_i\rangle\}}\sum\limits_{i}p_iN(|\phi_i\rangle),
\end{equation}
where the minimum is taken over all possible pure state decompositions $\{p_i,|\phi_i\rangle\}$ of $\rho_{AB}$.

For any bipartite pure state $|\phi\rangle_{AB}$, the negativity is given by $N(|\phi\rangle_{AB})=2\sum\limits_{i<j}\sqrt{\lambda_i\lambda_j}=(\Tr(\sqrt{\rho_A}))^2-1$, where $\lambda_i$ are the the eigenvalues of the reduced density matrix of $|\phi\rangle_{AB}$. For any bipartite pure state $|\phi\rangle_{AB}$ in $d\otimes d$ with Schmidt rank two, $|\phi\rangle_{AB}=\sqrt{\lambda_0}|00\rangle+\sqrt{\lambda_1}|11\rangle$, one has
\begin{equation}
N(|\phi\rangle_{AB})=\left\||\phi\rangle\langle\phi|^{T_B}\right\|-1=2\sqrt{\lambda_0\lambda_1}=\sqrt{2[1-\Tr(\rho_A^2)]}=C(|\phi\rangle_{AB}).
\end{equation}
In other words, the negativity is equivalent to the concurrence for any pure state with Schmidt rank two. Consequently it follows that for any two-qubit mixed state $\rho_{AB}=\sum\limits_{i}p_i|\phi_i\rangle\langle\phi_i|$,
\begin{eqnarray}
\widetilde{N}(\rho_{AB})&=&\min\limits_{\{p_i,|\phi_i\rangle\}}\sum\limits_{i}p_iN(|\phi_i\rangle)\nonumber\\
&=&\min\limits_{\{p_i,|\phi_i\rangle\}}\sum\limits_{i}p_iC(|\phi_i\rangle)\nonumber\\
&=&C(\rho_{AB}).
\end{eqnarray}

Recently, the monogamy relations satisfied by the $\alpha$th ($\alpha\geq2$) power of negativity for $N$-qubit systems have been studied \cite{25}. If $\widetilde{N}(\rho_{AB_i})\geq\widetilde{N}(\rho_{A|B_{i+1}\cdots B_{N-1}})$ for $i=1,2,\cdots,m$, and $\widetilde{N}(\rho_{AB_j})\leq\widetilde{N}(\rho_{A|B_{j+1}\cdots B_{N-1}})$ for $j=m+1,\cdots,N-2$, $1\leq m\leq N-3$, $N\geq4$, one has
\begin{eqnarray}
\widetilde{N}^{\alpha}(\rho_{A|B_1 \cdots B_{N-1}})&\geq& \widetilde{N}^{\alpha}(\rho_{AB_1})+(2^{\frac{\alpha}{2}}-1)\widetilde{N}^{\alpha}(\rho_{AB_2})+\cdots+(2^{\frac{\alpha}{2}}-1)^{m-1}\widetilde{N}^{\alpha}(\rho_{AB_m})\nonumber\\
& &+(2^{\frac{\alpha}{2}}-1)^{m+1}(\widetilde{N}^{\alpha}(\rho_{AB_{m+1}})+\cdots+\widetilde{N}^{\alpha}(\rho_{AB_{N-2}}))\nonumber\\
& &+(2^{\frac{\alpha}{2}}-1)^m\widetilde{N}^{\alpha}(\rho_{AB_{N-1}}).
\end{eqnarray}
This relation is further improved to be
\begin{eqnarray}
\widetilde{N}^{\alpha}(\rho_{A|B_1 \cdots B_{N-1}})&\geq&\widetilde{N}^{\alpha}(\rho_{AB_1})+(((1+k)^\frac{\alpha}{2}-1)/k^\frac{\alpha}{2})\widetilde{N}^{\alpha}(\rho_{AB_2})+\cdots\nonumber\\
& &+(((1+k)^\frac{\alpha}{2}-1)/k^\frac{\alpha}{2})^{m-1}\widetilde{N}^{\alpha}(\rho_{AB_m})\nonumber\\
& &+(((1+k)^\frac{\alpha}{2}-1)/k^\frac{\alpha}{2})^{m+1}(\widetilde{N}^{\alpha}(\rho_{AB_{m+1}})+\cdots+\widetilde{N}^{\alpha}(\rho_{AB_{N-2}}))\nonumber\\
& &+(((1+k)^\frac{\alpha}{2}-1)/k^\frac{\alpha}{2})^m\widetilde{N}^{\alpha}(\rho_{AB_{N-1}}),
\end{eqnarray}
with $k\widetilde{N}^2(\rho_{AB_i})\geq\widetilde{N}^2(\rho_{A|B_{i+1}\cdots B_{N-1}})$ for $i=1,2,\cdots,m$, and $\widetilde{N}^2(\rho_{AB_j})\leq k\widetilde{N}^2(\rho_{A|B_{j+1}\cdots B_{N-1}})$ for $j=m+1,\cdots,N-2$, $1\leq m\leq N-3$, $N\geq4$ and $0<k\leq1$ \cite{27} .

Similar to the consideration of concurrence, we have the following result. For convenience, we denote $\widetilde{N}_{AB_j}=\widetilde{N}(|\rho\rangle_{AB_j})$ for $j=1,2,\cdots,N-1$, and $\widetilde{N}_{A|B_1B_2\cdots B_{N-1}}=\widetilde{N}(\rho_{A|B_1 B_2\cdots B_{N-1}})$.

\begin{thm}
Let $\mu_r\geq1$ and $l_r\geq1$ ($1\le r\le N-2$) be real numbers.
For any $N$-qubit mixed state $\rho_{AB_1\cdots B_{N-1}}\in\mathcal{H}_A\otimes\mathcal{H}_{B_1}\otimes\cdots\otimes\mathcal{H}_{B_{N-1}}$, if $\widetilde{N}_{AB_i}^2\geq l_i\widetilde{N}_{A|B_{i+1}\cdots B_{N-1}}^2$, $\widetilde{N}_{A|B_i\cdots B_{N-1}}^2\geq \widetilde{N}_{AB_i}^2+\mu_i\widetilde{N}_{A|B_{i+1}\cdots B_{N-1}}^2$ for $i=1,2,\cdots,m$, and $\widetilde{N}_{A|B_{j+1}\cdots B_{N-1}}^2\geq l_j\widetilde{N}_{AB_j}^2$, $\widetilde{N}_{A|B_j\cdots B_{N-1}}^2\geq\mu_j\widetilde{N}_{AB_j}^2+\widetilde{N}_{A|B_{j+1}\cdots B_{N-1}}^2$ for $j=m+1,
\cdots,N-2$, $1\leq m\leq N-3,N\geq4$, then
\begin{eqnarray}
\widetilde{N}_{A|B_1\cdots B_{N-1}}^\alpha&\geq&\widetilde{N}_{AB_1}^\alpha+\mathcal{K}_1\widetilde{N}_{AB_2}^\alpha+\cdots+\mathcal{K}_1\cdots\mathcal{K}_{m-1}\widetilde{N}_{AB_m}^\alpha\nonumber\\
& &+\mathcal{K}_1\cdots\mathcal{K}_{m}(\mathcal{K}_{m+1}\widetilde{N}_{AB_{m+1}}^\alpha+\cdots+\mathcal{K}_{N-2}\widetilde{N}_{AB_{N-2}}^\alpha)\nonumber\\
& &+\mathcal{K}_1\cdots\mathcal{K}_{m}\widetilde{N}_{AB_{N-1}}^\alpha
\end{eqnarray}
for all $\alpha\geq2$, where $\mathcal{K}_r=(\mu_r+l_r)^{\frac{\alpha}{2}}-l_r^{\frac{\alpha}{2}}$ with $1\le r\le N-2$.
\end{thm}

In particular, we have

\begin{thm}
Let $\mu_r\geq1$ and $l_r\geq1$ ($1\le r\le N-2$) be real numbers.
For any $N$-qubit mixed state $\rho_{AB_1\cdots B_{N-1}}\in\mathcal{H}_A\otimes\mathcal{H}_{B_1}\otimes\cdots\otimes\mathcal{H}_{B_{N-1}}$, if $\widetilde{N}_{AB_i}^2\geq l_i\widetilde{N}_{A|B_{i+1}\cdots B_{N-1}}^2$, $\widetilde{N}_{A|B_i\cdots B_{N-1}}^2\geq \widetilde{N}_{AB_i}^2+\mu_i\widetilde{N}_{A|B_{i+1}\cdots B_{N-1}}^2$ for all $i=1,2,\cdots,N-2$, we have
\begin{eqnarray}
\widetilde{N}_{A|B_1\cdots B_{N-1}}^\alpha\geq\widetilde{N}_{AB_1}^\alpha+\mathcal{K}_1\widetilde{N}_{AB_2}^\alpha+\cdots
+\mathcal{K}_1\cdots\mathcal{K}_{N-2}\widetilde{N}_{AB_{N-1}}^\alpha
\end{eqnarray}
for all $\alpha\geq2$, where $\mathcal{K}_r=(\mu_r+l_r)^{\frac{\alpha}{2}}-l_r^{\frac{\alpha}{2}}$ with $1\le r\le N-2$.
\end{thm}

\begin{eg}
Let us consider the state in Example \ref{a} with $\lambda_0=\lambda_3=\lambda_4={1}/{\sqrt{5}}$ $\lambda_2=\sqrt{2/5}$ and $\lambda_1=0$. We have $\widetilde{N}_{A|BC}={4}/{5}$, $\widetilde{N}_{AB}={2\sqrt{2}}/{5}$ and $\widetilde{N}_{AC}={2}/{5}$. Therefore,
\begin{eqnarray}
\widetilde{N}_{AB}^\alpha+(2^\frac{\alpha}{2}-1)\widetilde{N}_{AC}^\alpha&=&({2\sqrt{2}}/{5})^\alpha+(2^\frac{\alpha}{2}-1)({2}/{5})^\alpha,\\
\widetilde{N}_{AB}^\alpha+(((1+k)^\frac{\alpha}{2}-1)/k^\frac{\alpha}{2})\widetilde{N}_{AC}^\alpha&=&({2\sqrt{2}}/{5})^\alpha+(((1+k)^\frac{\alpha}{2}-1)/k^\frac{\alpha}{2})({2}/{5})^\alpha,\\
\widetilde{N}_{AB}^\alpha+((\mu+l)^\frac{\alpha}{2}-l^\frac{\alpha}{2})\widetilde{N}_{AC}^\alpha&=&({2\sqrt{2}}/{5})^\alpha+((\mu+l)^\frac{\alpha}{2}-l^\frac{\alpha}{2})({2}/{5})^\alpha.
\end{eqnarray}
Our result is better than the one given in \cite{25,27} for $\alpha\geq2$, see Figure 3.
\end{eg}
\begin{figure}[h]
\begin{center}
\includegraphics[width=8cm,clip]{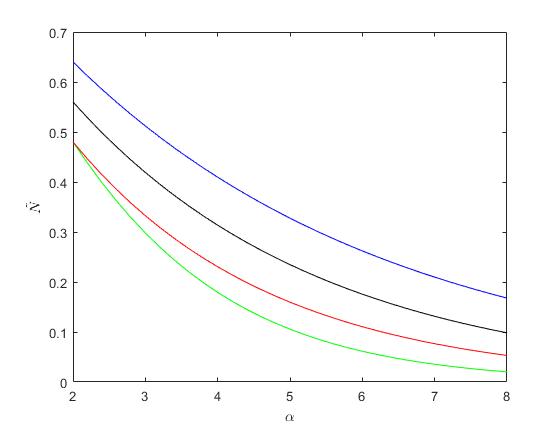}
\caption{From top to bottom, the first curve represents the negativity $\widetilde{N}(|\phi\rangle_{A|BC})$, the third and fourth curves represent the lower bounds from \cite{27} and \cite{25}, respectively, the second curve represents the lower bound from our result.}	
\end{center}
\end{figure}

\section{\bf Tighter monogamy and polygamy relations for Tsallis-q entanglement}
The Tsallis-q entanglement of a bipartite pure state $|\phi\rangle_{AB}$ is given by
\begin{equation}
T_q(|\phi\rangle_{AB})=S_q(\rho_A)=\frac{1}{q-1}(1-\Tr(\rho_A^q)),
\end{equation}
where $q\textgreater0$ and $q\not=1$ \cite{17}. $T_q(\rho)$ converges to the von Neumann entropy when $q$ tends to 1, $\lim\limits_{q\to1}T_q(\rho)=-\Tr\rho\log_2\rho=S(\rho)$. For a bipartite mixed state $\rho_{AB}$, the Tsallis-$q$ entanglement is defined as $T_q(\rho_{AB})=\min\limits_{\{p_i,|\phi_i\rangle\}}\sum\limits_ip_iT_q(|\phi_i\rangle)$, with the minimum taken over all possible pure state decompositions of $\rho_{AB}$. Yuan et al. presented an analytic relationship between the Tsallis-$q$ entanglement and concurrence for $\frac{5-\sqrt{13}}{2}\leq q\leq\frac{5+\sqrt{13}}{2}$,
\begin{equation}\label{p}
T_q(|\phi\rangle_{AB})=g_q(C^2(|\phi\rangle_{AB})),
\end{equation}
where $g_q(x)$ is defined as
\begin{equation}\label{q}
g_q(x)=\frac{1}{q-1}[1-(\frac{1+\sqrt{1-x}}{2})^q-(\frac{1-\sqrt{1-x}}{2})^q],
\end{equation}
with $0\leq x\leq1$ \cite{43}.
For a $2\otimes m$ pure state $|\phi\rangle$, it has been showed that $T_q(|\phi\rangle)=g_q(C^2(|\phi\rangle))$, and if $\rho$ is a two-qubit mixed state, then $T_q(\rho)=g_q(C^2(\rho))$ \cite{17}. Therefore, (\ref{p}) holds for any q such that $g_q(x)$ in (\ref{q}) is monotonically increasing and convex. Moreover, we have
$g_q(x^2+y^2)\geq g_q(x^2)+g_q(y^2)$ with $2\leq q\leq3$.

The Tsallis-$q$ entanglement satisfies the following relation,
\begin{equation}
T_{qA|B_1B_2\cdots B_{N-1}}\geq\sum\limits_{i=1}^{N-1}T_{qAB_i},
\end{equation}
where $i=1,2,\cdots,N-1$, $2\leq q\leq3$ \cite{17}. It is further proved in \cite{43} that
\begin{equation}
T^2_{qA|B_1B_2\cdots B_{N-1}}\geq\sum\limits_{i=1}^{N-1}T^2_{qAB_i},
\end{equation}
with $\frac{5-\sqrt{13}}{2}\leq q\leq\frac{5+\sqrt{13}}{2}$.

Recently, it has been proven that, for $N$-qubit mixed systems,
\begin{eqnarray}\label{r}
T_q^{\alpha}(\rho_{A|B_1 \cdots B_{N-1}})&\geq&T_q^{\alpha}(\rho_{AB_1})+(2^\alpha-1)T_q^{\alpha}(\rho_{AB_2})+\cdots+(2^\alpha-1)^{m-1}T_q^{\alpha}(\rho_{AB_m})\nonumber\\
& &+(2^\alpha-1)^{m+1}(T_q^{\alpha}(\rho_{AB_{m+1}})+\cdots+T_q^{\alpha}(\rho_{AB_{N-2}}))\nonumber\\
& &+(2^\alpha-1)^mT_q^{\alpha}(\rho_{AB_{N-1}}),
\end{eqnarray}
where $\alpha\geq1$, $2\leq q\leq3$, under the conditions that $C(\rho_{AB_i})\geq C(\rho_{A|B_{i+1}\cdots B_{N-1}})$ for $i=1,2,\cdots,m$, and $C(\rho_{AB_j})\leq C(\rho_{A|B_{j+1}\cdots B_{N-1}})$ for $j=m+1,\cdots,N-2$, $1\leq m\leq N-3$ and $N\geq4$ \cite{25}.
Later, the inequality (\ref{r}) is further improved as
\begin{eqnarray}
T_q^{\alpha}(\rho_{A|B_1 \cdots B_{N-1}})&\geq&T_q^{\alpha}(\rho_{AB_1})+(((1+k)^\alpha-1)/k^\alpha)T_q^{\alpha}(\rho_{AB_2})+\cdots\nonumber\\
& &+(((1+k)^\alpha-1)/k^\alpha)^{m-1}T_q^{\alpha}(\rho_{AB_m})\nonumber\\
& &+(((1+k)^\alpha-1)/k^\alpha)^{m+1}(T_q^{\alpha}(\rho_{AB_{m+1}})+\cdots+T_q^{\alpha}(\rho_{AB_{N-2}}))\nonumber\\
& &+(((1+k)^\alpha-1)/k^\alpha)^mT_q^{\alpha}(\rho_{AB_{N-1}}),
\end{eqnarray}
where $\alpha\geq1$, $2\leq q\leq3$, under the conditions that $kT_q(\rho_{AB_i})\geq T_q(\rho_{A|B_{i+1}\cdots B_{N-1}})$ for $i=1,2,\cdots,m$, and $T_q(\rho_{AB_j})\leq kT_q(\rho_{A|B_{j+1}\cdots B_{N-1}})$ for $j=m+1,\cdots,N-2$, $1\leq m\leq N-3$, $N\geq4$ and $0\textless k\leq1$ \cite{27}.

As a dual quantity to the Tsallis-$q$ entanglement, the Tsallis-$q$ entanglement of assistance (TEoA) is defined by $T_aq(\rho_{AB})=\max\limits_{\{p_i,|\phi_i\rangle\}}\sum\limits_ip_iT_q(|\phi_i\rangle)$, where the maximum is taken over all possible pure state decompositions
of $\rho_{AB}$ \cite{14}. If $1\leq q\leq2$ or $3\leq q\leq4$, the function $g_q$ defined in (\ref{q}) satisfies
\begin{equation}
g_q(x^2+y^2)\leq g_q(x^2)+g_q(y^2),
\end{equation}
which leads to the Tsallis polygamy inequality for any multi-qubit state $\rho_{AB_1B_2\cdots B_{N-1}}$ \cite{29},
\begin{equation}
T_{aqA|B_1B_2\cdots B_{N-1}}\leq\sum\limits_{i=1}^{N-1}T_{aqAB_i}.
\end{equation}

Taking a similar consideration to concurrence, we have the tighter monogamy and polygamy relations related to the Tsallis-$q$
entanglement as following. For convenience, we denote by $T_{qAB_j}=T_q(\rho_{AB_j})$ for $j=1,2,\cdots,N-1$, and $T_{qA|B_1B_2\cdots B_{N-1}}=T_q(\rho_{A|B_1 B_2\cdots B_{N-1}})$.

\begin{thm}
Let $\mu_r\geq1$ and $l_r\geq1$ ($1\le r\le N-2$) be real numbers. For any $N$-qubit mixed state $\rho_{AB_1\cdots B_{N-1}}\in\mathcal{H}_A\otimes\mathcal{H}_{B_1}\otimes\cdots\otimes\mathcal{H}_{B_{N-1}}$, if $T_{qAB_i}\geq l_iT_{qA|B_{i+1}\cdots B_{N-1}}$, $T_{qA|B_i\cdots B_{N-1}}\geq T_{qAB_i}+\mu_iT_{qA|B_{i+1}\cdots B_{N-1}}$ for $i=1,2,\cdots,m$, and $T_{qA|B_{j+1}\cdots B_{N-1}}\geq l_jT_{qAB_j}$, $T_{qA|B_j\cdots B_{N-1}}\geq\mu_jT_{qAB_j}+T_{qA|B_{j+1}\cdots B_{N-1}}$ for $j=m+1,
\cdots,N-2$, $1\leq m\leq N-3,N\geq4$, then
\begin{eqnarray}
T_{qA|B_1\cdots B_{N-1}}^\alpha&\geq&T_{qAB_1}^\alpha+\mathcal{K}_1T_{qAB_2}^\alpha+\cdots+\mathcal{K}_1\cdots\mathcal{K}_{m-1}T_{qAB_m}^\alpha\nonumber\\
& &+\mathcal{K}_1\cdots\mathcal{K}_{m}(\mathcal{K}_{m+1}T_{qAB_{m+1}}^\alpha+\cdots+\mathcal{K}_{N-2}T_{qAB_{N-2}}^\alpha)\nonumber\\
& &+\mathcal{K}_1\cdots\mathcal{K}_{m}T_{qAB_{N-1}}^\alpha
\end{eqnarray}
for all $\alpha\geq1$ and $2\leq q\leq3$, where $\mathcal{K}_r=(\mu_r+l_r)^\alpha-l_r^\alpha$ with $1\le r\le N-2$.
\end{thm}

The above theorem gives rise to, in particular,

\begin{thm}
Let $\mu_r\geq1$ and $l_r\geq1$ ($1\le r\le N-2$) be real numbers. For any $N$-qubit mixed state $\rho_{AB_1\cdots B_{N-1}}\in\mathcal{H}_A\otimes\mathcal{H}_{B_1}\otimes\cdots\otimes\mathcal{H}_{B_{N-1}}$, if $T_{qAB_i}\geq l_iT_{qA|B_{i+1}\cdots B_{N-1}}$, $T_{qA|B_i\cdots B_{N-1}}\geq T_{qAB_i}+\mu_iT_{qA|B_{i+1}\cdots B_{N-1}}$ for all $i=1,2,\cdots,N-2$, then
\begin{eqnarray}
T_{qA|B_1\cdots B_{N-1}}^\alpha\geq T_{qAB_1}^\alpha+\mathcal{K}_1T_{qAB_2}^\alpha+\cdots
+\mathcal{K}_1\cdots\mathcal{K}_{N-2}T_{qAB_{N-1}}^\alpha
\end{eqnarray}
for all $\alpha\geq1$ and $2\leq q\leq3$, where $\mathcal{K}_r=(\mu_r+l_r)^\alpha-l_r^\alpha$ with $1\le r\le N-2$.
\end{thm}

\begin{eg}
Let us consider the state in Example \ref{a} with $\lambda_0=\lambda_3=\lambda_4={1}/{\sqrt{5}}$, $\lambda_2=\sqrt{{2}/{5}}$ and $\lambda_1=0$. For $q=2$, we have $T_{2A|BC}={8}/{25},T_{2AB}={4}/{25}$, and $T_{2AC}={2}/{25}$. Then
\begin{eqnarray}
T_{2AB}^\alpha+(2^\alpha-1)T_{2AC}^\alpha&=&({4}/{25})^\alpha+(2^\alpha-1)({2}/{25})^\alpha,\\
T_{2AB}^\alpha+(((1+k)^\alpha-1)/k^\alpha)T_{2AC}^\alpha&=&({4}/{25})^\alpha+(((1+k)^\alpha-1)/k^\alpha)({2}/{25})^\alpha,\\
T_{2AB}^\alpha+((\mu+l)^\alpha-l^\alpha)T_{2AC}^\alpha&=&({4}/{25})^\alpha+((\mu+l)^\alpha-l^\alpha)({2}/{25})^\alpha.
\end{eqnarray}
We see that our result is better than the one given in \cite{25,27} for $\alpha\geq1$, see Figure 4.
\end{eg}
\begin{figure}[h]
\begin{center}
\includegraphics[width=8cm,clip]{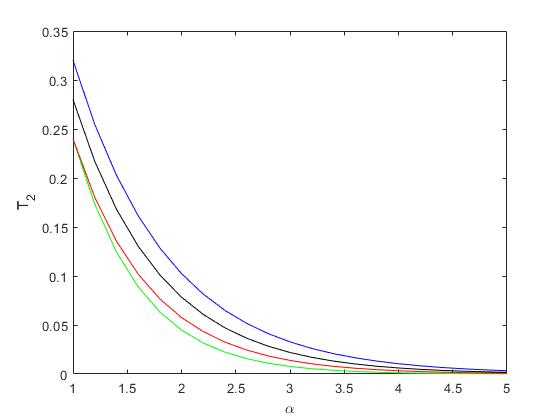}
\caption{From top to bottom, the first curve represents the Tsallis-$q$ entanglement $T_q(|\phi\rangle_{A|BC})$, the third and fourth curves represent the lower bounds from \cite{27} and \cite{25}, respectively, the second curve represents the lower bound from our result.}	
\end{center}
\end{figure}

For the Tsallis-$q$ entanglement of assistance (TEoA), we have

\begin{thm}
 Let $0\textless\mu_r\leq1$ and $l_r\geq1$ ($1\le r\le N-2$) be real numbers.For any $N$-qubit mixed state $\rho_{AB_1\cdots B_{N-1}}\in\mathcal{H}_A\otimes\mathcal{H}_{B_1}\otimes\cdots\otimes\mathcal{H}_{B_{N-1}}$, if $T_{aqAB_i}\geq l_iT_{aqA|B_{i+1}\cdots B_{N-1}}$, $T_{aqA|B_i\cdots B_{N-1}}\leq T_{aqAB_i}+\mu_iT_{aqA|B_{i+1}\cdots B_{N-1}}$ for $i=1,2,\cdots,m$, and $T_{aqA|B_{j+1}\cdots B_{N-1}}\geq l_jT_{aqAB_j}$, $T_{aqA|B_j\cdots B_{N-1}}\leq\mu_jT_{aqAB_j}+T_{aqA|B_{j+1}\cdots B_{N-1}}$ for $j=m+1,
\cdots,N-2$, $1\leq m\leq N-3,N\geq4$, we have
\begin{eqnarray}
T^\alpha_{aqA|B_1\cdots B_{N-1}}&\leq&T_{aqAB_1}^\alpha+\mathcal{K}_1T_{aqAB_2}^\alpha+\cdots+\mathcal{K}_1\cdots\mathcal{K}_{m-1}T_{aqAB_m}^\alpha\nonumber\\
& &+\mathcal{K}_1\cdots\mathcal{K}_{m}(\mathcal{K}_{m+1}T_{aqAB_{m+1}}^\alpha+\cdots+\mathcal{K}_{N-2}T_{aqAB_{N-2}}^\alpha)\nonumber\\
& &+\mathcal{K}_1\cdots\mathcal{K}_{m}T_{aqAB_{N-1}}^\alpha
\end{eqnarray}
for all $0\leq\alpha\leq1$ with $1\leq q\leq2$ and $3\leq q\leq4$, where $\mathcal{K}_r=(\mu_r+l_r)^{\alpha}-l_r^{\alpha}$, $1\le r\le N-2$.
\end{thm}

Particularly, one has

\begin{thm}
Let $0\textless\mu_r\leq1$ and $l_r\geq1$ ($1\le r\le N-2$) be real numbers. For any $N$-qubit mixed state $\rho_{AB_1\cdots B_{N-1}}\in\mathcal{H}_A\otimes\mathcal{H}_{B_1}\otimes\cdots\otimes\mathcal{H}_{B_{N-1}}$, if $T_{aqAB_i}\geq l_iT_{aqA|B_{i+1}\cdots B_{N-1}}$, $T_{aqA|B_i\cdots B_{N-1}}\leq T_{aqAB_i}+\mu_iT_{aqA|B_{i+1}\cdots B_{N-1}}$ for all $i=1,2,\cdots,N-2$, then
\begin{eqnarray}
T_{aqA|B_1\cdots B_{N-1}}^\alpha\leq T_{aqAB_1}^\alpha+\mathcal{K}_1T_{aqAB_2}^\alpha+\cdots
+\mathcal{K}_1\cdots\mathcal{K}_{N-2}T_{aqAB_{N-1}}^\alpha
\end{eqnarray}
for all $0\leq\alpha\leq1$ with $1\leq q\leq2$ and $3\leq q\leq4$, where $\mathcal{K}_r=(\mu_r+l_r)^{\alpha}-l_r^{\alpha}$, $1\le r\le N-2$.
\end{thm}

\section{\bf Tighter monogamy and polygamy relations for R$\acute{e}$nyi-$\alpha$ entanglement}
For a bipartite pure state $|\phi\rangle_{AB}$, the R$\acute{e}$nyi-$\alpha$ entanglement is defined as $E_{\acute\alpha}(|\phi\rangle_{AB})=S_{\acute\alpha}(\rho_A)$, where $S_{\acute\alpha}(\rho)=\frac{1}{1-\acute{\alpha}}\log_2(\Tr\rho^{\acute\alpha})$ for any $\acute{\alpha}\textgreater0$ and $\acute{\alpha}\not=1$, and $\lim\limits_{\acute{\alpha}\to1}S_{\acute\alpha}(\rho)=S(\rho)=-\Tr\rho\log_2\rho$ \cite{44}. For a bipartite mixed state $\rho_{AB}$, the R$\acute{e}$nyi-$\alpha$ entanglement is given by $E_{\acute\alpha}(\rho_{AB})=\min\limits_{\{p_i,|\phi_i\rangle\}}\sum\limits_ip_iE_{\acute\alpha}(|\phi_i\rangle)$, where the minimum is taken over all possible pure state decompositions of $\rho_{AB}$. For each $\acute{\alpha}\textgreater0$, one has $E_{\acute\alpha}(\rho_{AB})=f_{\acute\alpha}(C(\rho_{AB}))$, where $f_{\acute\alpha}(x)=\frac{1}{1-\acute\alpha}\log[(\frac{1-\sqrt{1-x^2}}{2})^2+(\frac{1+\sqrt{1-x^2}}{2})^2]$ is a monotonically increasing and convex function \cite{28}. For $\acute{\alpha}\geq2$ and any $N$-qubit state $\rho_{AB_1B_2\cdots B_{N-1}}$, one has $E_{\acute\alpha}(\rho_{A|B_1B_2\cdots B_{N-1}})\geq E_{\acute\alpha}(\rho_{AB_1})+E_{\acute\alpha}(\rho_{A|B_2})+\cdots+E_{\acute\alpha}(\rho_{A|B_{N-1}})$ \cite{17}.

The R$\acute{e}$nyi-$\alpha$ entanglement of assistance (REoA), a dual quantity to the R$\acute{e}$nyi-$\alpha$ entanglement, is defined as $E_{a\acute{\alpha}}(\rho_{AB})=\max\limits_{\{p_i,|\phi_i\rangle\}}\sum\limits_ip_iE_{\acute\alpha}(|\phi_i\rangle)$, where the maximum is taken over all possible pure state decompositions of $\rho_{AB}$. For $\acute{\alpha}\in[\frac{\sqrt{7}-1}{2},\frac{\sqrt{13}-1}{2}]$ and any $N$-qubit state $\rho_{AB_1B_2\cdots B_{N-1}}$, a polygamy relation of multi-partite quantum entanglement in terms of REoA has been presented \cite{23}, $E_{a\acute{\alpha}}(\rho_{A|B_1B_2\cdots B_{N-1}})\leq E_{a\acute{\alpha}}(\rho_{AB_1})+E_{a\acute{\alpha}}(\rho_{A|B_2})+\cdots+E_{a\acute{\alpha}}(\rho_{A|B_{N-1}})$.

We propose the following monogamy and polygamy relations for the R$\acute{e}$nyi-$\alpha$ entanglement, which are tighter than the previous results. For convenience, we denote by $E_{\acute{\alpha}AB_j}=E_{\acute{\alpha}}(\rho_{AB_j})$ for $j=1,2,\cdots,N-1$, and $E_{\acute{\alpha}A|B_1B_2\cdots B_{N-1}}=E_{\acute{\alpha}}(\rho_{A|B_1 B_2\cdots B_{N-1}})$.

\begin{thm}
Let $\mu_r\geq1$ and $l_r\geq1$ ($1\le r\le N-2$) be real numbers. For any $N$-qubit mixed state $\rho_{AB_1\cdots B_{N-1}}\in\mathcal{H}_A\otimes\mathcal{H}_{B_1}\otimes\cdots\otimes\mathcal{H}_{B_{N-1}}$, if $E_{\acute{\alpha}AB_i}\geq l_iE_{\acute{\alpha}A|B_{i+1}\cdots B_{N-1}}$, $E_{\acute{\alpha}A|B_i\cdots B_{N-1}}\geq E_{\acute{\alpha}AB_i}+\mu_iE_{\acute{\alpha}A|B_{i+1}\cdots B_{N-1}}$ for $i=1,2,\cdots,m$, and $E_{\acute{\alpha}A|B_{j+1}\cdots B_{N-1}}\geq l_jE_{\acute{\alpha}AB_j}$, $E_{\acute{\alpha}A|B_j\cdots B_{N-1}}\geq\mu_jE_{\acute{\alpha}AB_j}+E_{\acute{\alpha}A|B_{j+1}\cdots B_{N-1}}$ for $j=m+1,
\cdots,N-2$, $1\leq m\leq N-3,N\geq4$, we have
\begin{eqnarray}
E_{\acute{\alpha}A|B_1\cdots B_{N-1}}^\alpha&\geq& E_{\acute{\alpha}AB_1}^\alpha+\mathcal{K}_1E_{\acute{\alpha}AB_2}^\alpha+\cdots+\mathcal{K}_1\cdots\mathcal{K}_{m-1}E_{\acute{\alpha}AB_m}^\alpha\nonumber\\
& &+\mathcal{K}_1\cdots\mathcal{K}_{m}(\mathcal{K}_{m+1}E_{\acute{\alpha}AB_{m+1}}^\alpha+\cdots+\mathcal{K}_{N-2}E_{\acute{\alpha}AB_{N-2}}^\alpha)\nonumber\\
& &+\mathcal{K}_1\cdots\mathcal{K}_{m}E_{\acute{\alpha}AB_{N-1}}^\alpha
\end{eqnarray}
for all $\alpha\geq1$ and $\acute{\alpha}\geq2$, where $\mathcal{K}_r=(\mu_r+l_r)^{\alpha}-l_r^{\alpha}$, $1\le r\le N-2$.
\end{thm}

\begin{thm}
Let $\mu_r\geq1$ and $l_r\geq1$ ($1\le r\le N-2$) be real numbers. For any $N$-qubit mixed state $\rho_{AB_1\cdots B_{N-1}}\in\mathcal{H}_A\otimes\mathcal{H}_{B_1}\otimes\cdots\otimes\mathcal{H}_{B_{N-1}}$, if $E_{\acute{\alpha}AB_i}\geq l_iE_{\acute{\alpha}A|B_{i+1}\cdots B_{N-1}}$, $E_{\acute{\alpha}A|B_i\cdots B_{N-1}}\geq E_{\acute{\alpha}AB_i}+\mu_iE_{\acute{\alpha}A|B_{i+1}\cdots B_{N-1}}$ for all $i=1,2,\cdots,N-2$, then
\begin{eqnarray}
E_{\acute{\alpha}A|B_1\cdots B_{N-1}}^\alpha\geq E_{\acute{\alpha}AB_1}^\alpha+\mathcal{K}_1E_{\acute{\alpha}AB_2}^\alpha+\cdots
+\mathcal{K}_1\cdots\mathcal{K}_{N-2}E_{\acute{\alpha}AB_{N-1}}^\alpha
\end{eqnarray}
for all $\alpha\geq1$ and $\acute{\alpha}\geq2$, where $\mathcal{K}_r=(\mu_r+l_r)^{\alpha}-l_r^{\alpha}$, $1\le r\le N-2$.
\end{thm}

\begin{eg}
Let us consider the state in Example \ref{a} with $\lambda_0=\lambda_3=\lambda_4={1}/{\sqrt{5}}$, $\lambda_2=\sqrt{{2}/{5}}$ and $\lambda_1=0$. For $\acute{\alpha}=2$, we have $E_{2A|BC}=\log_2({25}/{17})\approx0.556393$, $E_{2AB}=\log_2({25}/{21})\approx0.251539$ and $E_{2AC}=\log_2({25}/{23})\approx0.120294$. Then
\begin{eqnarray}
E_{2AB}^\alpha+E_{2AC}^\alpha&=&(0.251539)^\alpha+(0.120294)^\alpha,\\
E_{2AB}^\alpha+(((1+k)^\alpha-1)/k^\alpha)E_{2AC}^\alpha&=&(0.251539)^\alpha+(((1+k)^\alpha-1)/k^\alpha)(0.120294)^\alpha,\\
E_{2AB}^\alpha+((\mu+l)^\alpha-l^\alpha)E_{2AC}^\alpha&=&(0.251539)^\alpha+((\mu+l)^\alpha-l^\alpha)(0.120294)^\alpha,
\end{eqnarray}
which show that our result is better than the one given in \cite{25,27} for $\alpha\geq1$, see Figure 5.
\end{eg}
\begin{figure}[h]
\begin{center}
\includegraphics[width=8cm,clip]{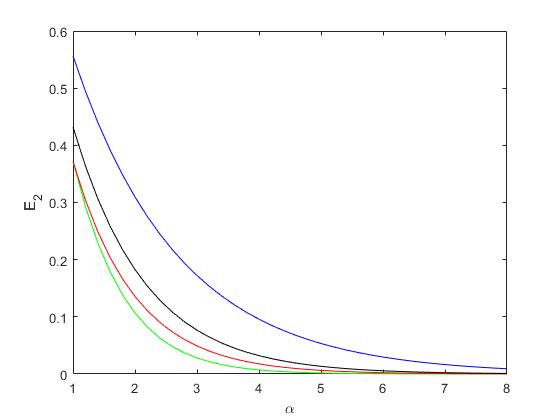}
\caption{From top to bottom, the first curve represents the R$\acute{e}$nyi-$\alpha$ entanglement $E_{\acute{\alpha}}(|\phi\rangle_{A|BC})$, the third and fourth curves represent the lower bounds from \cite{27} and \cite{25}, respectively, the second curve represents the lower bound from our result.}	
\end{center}
\end{figure}

Correspondingly, for $E_{a\acute{\alpha}}$ we have

\begin{thm}
Let $0 \textless\mu_r\leq1$ and $l_r\geq1$ ($1\le r\le N-2$) be real numbers. For any $N$-qubit mixed state $\rho_{AB_1\cdots B_{N-1}}\in\mathcal{H}_A\otimes\mathcal{H}_{B_1}\otimes\cdots\otimes\mathcal{H}_{B_{N-1}}$, if $E_{a\acute{\alpha}AB_i}\geq l_iE_{a\acute{\alpha}A|B_{i+1}\cdots B_{N-1}}$, $E_{a\acute{\alpha}A|B_i\cdots B_{N-1}}\leq E_{a\acute{\alpha}AB_i}+\mu_iE_{a\acute{\alpha}A|B_{i+1}\cdots B_{N-1}}$ for $i=1,2,\cdots,m$, and $E_{a\acute{\alpha}A|B_{j+1}\cdots B_{N-1}}\geq l_jE_{a\acute{\alpha}AB_j}$, $E_{a\acute{\alpha}A|B_j\cdots B_{N-1}}\leq\mu_jE_{a\acute{\alpha}AB_j}+E_{a\acute{\alpha}A|B_{j+1}\cdots B_{N-1}}$ for $j=m+1,
\cdots,N-2$, $1\leq m\leq N-3,N\geq4$, we have
\begin{eqnarray}
E_{a\acute{\alpha}A|B_1\cdots B_{N-1}}^\alpha&\leq&E_{a\acute{\alpha}AB_1}^\alpha+\mathcal{K}_1E_{a\acute{\alpha}AB_2}^\alpha+\cdots+\mathcal{K}_1\cdots\mathcal{K}_{m-1}E_{a\acute{\alpha}AB_m}^\alpha\nonumber\\
& &+\mathcal{K}_1\cdots\mathcal{K}_{m}(\mathcal{K}_{m+1}E_{a\acute{\alpha}AB_{m+1}}^\alpha+\cdots+\mathcal{K}_{N-2}E_{a\acute{\alpha}AB_{N-2}}^\alpha)\nonumber\\
& &+\mathcal{K}_1\cdots\mathcal{K}_{m}E_{a\acute{\alpha}AB_{N-1}}^\alpha
\end{eqnarray}
for all $0\leq\alpha\leq1$ and $\frac{\sqrt{7}-1}{2}\leq\acute{\alpha}\leq\frac{\sqrt{13}-1}{2}$, where $\mathcal{K}_r=(\mu_r+l_r)^{\alpha}-l_r^{\alpha}$, $1\le r\le N-2$.
\end{thm}

\begin{thm}
Let $0\textless\mu_r\leq1$ and $l_r\geq1$ ($1\le r\le N-2$) be real numbers. For any $N$-qubit mixed state $\rho_{AB_1\cdots B_{N-1}}\in\mathcal{H}_A\otimes\mathcal{H}_{B_1}\otimes\cdots\otimes\mathcal{H}_{B_{N-1}}$, if $E_{a\acute{\alpha}AB_i}\geq l_iE_{a\acute{\alpha}A|B_{i+1}\cdots B_{N-1}}$, $E_{a\acute{\alpha}A|B_i\cdots B_{N-1}}\leq E_{a\acute{\alpha}AB_i}+\mu_iE_{a\acute{\alpha}A|B_{i+1}\cdots B_{N-1}}$ for all $i=1,2,\cdots,N-2$, then
\begin{eqnarray}
E_{a\acute{\alpha}A|B_1\cdots B_{N-1}}^\alpha\leq E_{a\acute{\alpha}AB_1}^\alpha+\mathcal{K}_1E_{a\acute{\alpha}AB_2}^\alpha+\cdots
+\mathcal{K}_1\cdots\mathcal{K}_{N-2}E_{a\acute{\alpha}AB_{N-1}}^\alpha,
\end{eqnarray}
for all $0\leq\alpha\leq1$ and $\frac{\sqrt{7}-1}{2}\leq\acute{\alpha}\leq\frac{\sqrt{13}-1}{2}$, where $\mathcal{K}_r=(\mu_r+l_r)^{\alpha}-l_r^{\alpha}$, $1\le r\le N-2$.
\end{thm}

\section{\bf Conclusion}
 We have provided tighter monogamy inequalities with respect to the concurrence, entanglement of formation, convex-roof extended negativity, Tsallis-q entanglement and R$\acute{e}$nyi-$\alpha$ entanglement, we have also provided tighter polygamy inequalities with respect to the entanglement of formation, Tsallis-q entanglement and R$\acute{e}$nyi-$\alpha$ entanglement. Monogamy and polygamy inequalities play significant roles in characterizing the entanglement distributions and shareability in multipartite quantum systems. Tighter monogamy relations imply finer characterizations of the entanglement distribution. Our approach may also be used to study the monogamy properties related to other quantum correlations, and provides a useful way to understand the property of multipartite entanglement.

\vspace{5pt}

\noindent {\bf Acknowledgments}
This work is supported by the NSF of China under Grant No. 12075159, Beijing Natural Science Foundation (Z190005), Key Project of Beijing Municipal Commission of Education (KZ201810028042), the Academician Innovation Platform of Hainan Province, Academy for Multidisciplinary Studies, Capital Normal University, and Shenzhen Institute for Quantum Science and Engineering, Southern University of Science and Technology (Grant No. SIQSE202001).

%\begin{thebibliography}{9}

\end{document}